\documentclass[aps,pra,a4paper,twocolumn,superscriptaddress,floatfix]{revtex4}

\usepackage{amssymb}
\usepackage{graphicx}
\usepackage{amsmath}
\usepackage{amsthm}
\usepackage{amsmath}
\usepackage{amsfonts}
\usepackage{mathtools}
\usepackage{latexsym}
\usepackage{dsfont}
\usepackage{marvosym}
\usepackage{mathrsfs}
\usepackage[utf8]{inputenc}
\usepackage[english]{babel}
\usepackage{bbm}
\usepackage{nonfloat}
\usepackage{color}
\usepackage{bm}
\newcommand{\Tr}{\mathop{\text{Tr}}\nolimits}
\renewcommand{\Re}{\mathop{\text{Re}}\nolimits}
\renewcommand{\Im}{\mathop{\text{Im}}\nolimits}
\newcommand{\Diag}{\mathop{\mathcal{D}iag}\nolimits}
\newcommand{\rank}{\mathop{\text{rank}}\nolimits}
\newcommand{\Span}{\mathop{\text{span}}\nolimits}
\theoremstyle{plain}
\newtheorem{theo}{Theorem}
\newtheorem{lem}{Lemma}
\newtheorem{cor}{Corollary}
\newtheorem{mydef}{Definition}
\newtheorem{prop}{Proposition}

\definecolor{dgreen}{rgb}{0,0.5,0}
\definecolor{delete}{cmyk}{0.5,0,0,0}

\begin{document}
\title{Hamiltonian Purification}

\author{Davide Orsucci}
\affiliation{Scuola Normale Superiore, I-56126 Pisa, Italy.}

\author{Daniel Burgarth}
\affiliation{Department of Mathematics, Aberystwyth University, Aberystwyth SY23 3BZ, United Kingdom}

\author{Paolo Facchi}
\affiliation{Dipartimento di Fisica and MECENAS, Universit\`a di Bari, I-70126 Bari, Italy}
\affiliation{INFN, Sezione di Bari, I-70126 Bari, Italy}

\author{Hiromichi Nakazato}
\affiliation{Department of Physics, Waseda University, Tokyo 169-8555, Japan}

\author{Saverio Pascazio}
\affiliation{Dipartimento di Fisica and MECENAS, Universit\`a di Bari, I-70126 Bari, Italy}
\affiliation{INFN, Sezione di Bari, I-70126 Bari, Italy}

\author{Kazuya Yuasa}
\affiliation{Department of Physics, Waseda University, Tokyo 169-8555, Japan}

\author{Vittorio Giovannetti}
\affiliation{NEST, Scuola Normale Superiore and Istituto Nanoscienze-CNR, I-56126 Pisa, Italy}

\begin{abstract}
The problem of Hamiltonian purification introduced by Burgarth \textit{et~al.}\ [D.~K. Burgarth \textit{et~al.}, Nat.\ Commun.\ \textbf{5},  5173  (2014)] is formalized and discussed.
Specifically, given a set of non-commuting Hamiltonians $\{ h_1, \ldots, h_m\}$ operating on a $d$-dimensional quantum system $\mathscr{H}_d$,
the problem consists in identifying a set of commuting Hamiltonians $\{ H_1, \ldots, H_m\}$ operating on a larger $d_E$-dimensional system $\mathscr{H}_{d_E}$ which embeds $\mathscr{H}_d$ as a proper subspace, such that $h_j = PH_jP$ with $P$ being the projection which allows one to recover $\mathscr{H}_d$ from $\mathscr{H}_{d_E}$. The notions of spanning-set purification and generator purification of an algebra are also introduced and optimal solutions for  $\mathfrak{u}(d)$ are provided.
 \end{abstract}

\maketitle

\section{Introduction}\label{sec:int}
The possibility of achieving the control over a quantum system  is the fundamental prerequisite for developing a new form of technology based on quantum effects \cite{QTECH,QTECH1,QTECH2}.  In particular this is an essential requirement for quantum computation, quantum communication, and more generally for all other data processing procedures that involve quantum systems as information carriers \cite{NC00}.

In many experimental settings, quantum control is implemented via an  electromagnetic field interacting with the system of interest, as happens for \textit{cold atoms} in optical lattices \cite{nca}, for \textit{trapped ions} \cite{ion}, for \textit{electrons} in quantum dots \cite{dots}, and actually in virtually all experiments in low energy physics. In this context  the electromagnetic field can be often treated as a classical field (in the limit of many \textit{quanta}), allowing a semiclassical description of the control over the quantum system \cite{Ale07,DP10,DH08}.
 Furthermore in many cases of physical interest the whole process can be effectively formalized by assuming that via proper manipulation of the field parameters the experimenter produces 
 a series of pulses implementing some specially engineered control Hamiltonians from a discrete set $\{H_1,\ldots, H_m\}$. Such pulses are assumed to be applied in any order, for any durations, by switching them on and off very sharply, 
the resulting transformation being a unitary evolution of the form
$e^{-iH_{j_N}t_N} \cdots e^{-iH_{j_1}t_1}$ 
with $j_1,\ldots,j_N \in \{1,\ldots,m\}$ and $t_1, \ldots, t_N$ being the  selected temporal durations (hereafter $\hbar$ is set to unity for simplicity) \cite{NOTA1}. By the \textit{Lie-algebraic rank condition} \cite{Ale07} the unitary operators that can be realized 
via such procedure
  are those in the connected Lie group associated to the real Lie algebra $\mathfrak{Lie}(H_1 , \ldots , H_m )$ generated by the Hamiltonians $\{H_1,\ldots, H_m\}$ [where $\mathfrak{Lie}(H_1 , \ldots , H_m )$ is formed by the real linear combinations of 
 $H_j$ and their iterated commutators $i[H_{j_1},  H_{j_2}]$,  $i\bm{[}H_{j_1}, i[H_{j_2},  H_{j_3}]\bm{]}$, etc.], i.e.,  $e^{-i \Xi}$ with $\Xi \in \mathfrak{Lie}(H_1 , \ldots , H_m )$. 
In this framework one then says that  full (unitary) controllability is achieved if the dimension of   $\mathfrak{Lie}(H_1 , \ldots , H_m )$ is large enough to permit the implementation of all possible unitary transformations on the system, i.e.,  if  $\mathfrak{Lie}(H_1 , \ldots , H_m )$ coincides  with the 
complete algebra $\mathfrak{u}(d)$ formed by self-adjoint $d\times d$ complex matrices \cite{NOTA2}, $d$ being the dimension of the controlled system.

The above scheme is the paradigmatic example of what is typically identified as \textit{open-loop} or \textit{non-adaptive} control, where all the operations are completely determined prior to the control experiment \cite{Ale07,DP10}. In other words the  system is driven in the absence  of an external feedback loop, i.e., without using any information gathered (via measurement) 
\textit{during} the evolution.
It turns out that in quantum mechanics an alternative mechanism of non-adaptive control is available: it is enforced via quantum Zeno dynamics \cite{FP02,ZenoMP}. In this scenario, while measurements are present, the associated outcomes are not used to guide the forthcoming operations: only their effects on the system evolution are exploited (a fact which has no analog in the classical domain).  The underlying physical principle is the following. When a quantum system undergoes a sharp  (von Neumann)  measurement, it is projected into one of the associated eigenspaces of the observable, say the space $\mathscr{H}_P$ characterized by an orthogonal projection $P$. It is then let to undergo a unitary evolution $e^{-i H \Delta t}$  for a short time $\Delta t$ and is measured again via the same von Neumann measurement. The probability to find it in a different measurement eigenspace $\mathscr{H}_{P'}$ orthogonal to the original one $\mathscr{H}_{P}$ is proportional to $(\Delta t)^2$. Instead, with high probability, the system remains in $\mathscr{H}_{P}$, while experiencing an effective unitary rotation of the form  $e^{-i h \Delta  t}$ induced by the projected Hamiltonian  $h := PHP$ \cite{FP02,ZenoMP,ZENO}. 
Accordingly, in the limit of infinitely frequent measurements performed within a fixed time interval $t$, the system remains in the subspace $\mathscr{H}_P$, evolving through an effective \textit{Zeno dynamics} described by the operator 
\begin{equation} \label{zenosequence}
	\lim_{N\to\infty}
	(P e^{-i  H {t}/{N}}  P)^N = P e^{ -i PHP t} = P e^{-i h t}.
\end{equation}

In Ref.\ \cite{Bur14}, it was shown that, by adopting the quantum Zeno dynamics, the control that the experimenter can enforce on a quantum system  can be greatly enhanced. For example, consider the case where the set of engineered Hamiltonians contains only two commuting elements $H_1$ and $H_2$. The associated Lie algebra they generate is just two-dimensional and hence is not sufficient to induce full controllability, even for the smallest quantum system, a qubit --- indeed $\dim[\mathfrak{u}(d=2)] =4$. 
Under these conditions, it turns out that for a proper choice of the projection $P$ it may happen that the projected counterparts $h_1=PH_1P$
and $h_2=PH_2P$  of the control Hamiltonians do not commute. Accordingly 
the Lie algebra generated by $\{h_1,h_2\}$ can be much larger than the one 
associated with $\{H_1, H_2\}$, 
and consequently the control exerted much finer.
In particular  the enhancement can be exponential in the system size. For instance in Ref.\ \cite{Bur14}  an explicit example is given where two commuting Hamiltonians   $H_1$ and $H_2$ act  on a chain of $n$ qubits, and once a proper Zeno projection $P$ is  applied on the first qubit of the chain the resulting 
Zeno Hamiltonians $h_1$ and $h_2$  generate the full algebra of traceless Hermitian operators acting on the remaining $n-1$ qubits (which is a Lie algebra of dimension $4^{n-1}$), thus allowing to perform any unitary operations on them.
Moreover, it can be shown that this is indeed a quite general phenomenon. In fact a simple argument \cite{Bur14}  shows that if a system is controllable for a specific choice of the parameters, then it is controllable for almost all choices of the parameters (with respect, e.g., to the Lebesgue measure). In the present case it means that, for almost all choices of a rank-$2^{n-1}$ projection $P$ and of two commuting Hamiltonians $\{{H}_1,{H}_2\}$, the system is fully controllable in the projected subspace $\mathscr{H}_{{P}}$ with the Hamiltonians
${h}_1 = {P} {H}_1 {P}$ and ${h}_2 = {P} {H}_2 {P}$.

The aforementioned results of \cite{Bur14}  show that as few as two commuting Hamiltonians, when projected on a smaller subspace of dimension $d$ through the Zeno mechanism, may achieve to generate the whole Lie algebra $\mathfrak{u}(d)$. The scope of the present article is to investigate the opposite question: given a set of Hamiltonians $\{h_1, \ldots, h_m\}$, which are non-commuting in general,  is it possible to extend them to a set of commuting Hamiltonians  $\{ H_1, \ldots, H_m\}$ from which $h_j$ can be obtained via a proper projection of the latter (i.e., $h_j = P H_jP$)? 
We call this operation \textit{Hamiltonian purification}, taking inspiration from similar problems which have been investigated in quantum information. For instance, we recall that by the \textit{state purification} \cite{NC00} a quantum mixed state $\rho$ on a system $S \cong \mathscr{H}_d$ is extended to a pure state $|\psi_\rho\rangle$ on a system $S+A \cong \mathscr{H}_d \otimes \mathscr{H}_d$, from which $\rho$ can be recovered through a partial trace over the ancilla system $A \cong \mathscr{H}_d$. Another similar result can be obtained for the \textit{channel purification (Stinespring dilation theorem)} or for the \textit{purification of positive operator-valued measure (POVM) (Naimark extension theorem)}, according to which all the completely positive trace-preserving linear maps and all the generalized measurement procedures, respectively, can be described as unitary transformations
on an extended system followed by partial trace~\cite{NC00,REV}.

In what follows  we start  by presenting a formal characterization of the Hamiltonian purification problem and of the associated notions of {\it spanning-set purification} and {\it generator purification} of an algebra (see Sec.\ \ref{sec1}). Then we prove some theorems regarding the minimal dimension $d_E^{\text{(min)}}$ of the extended Hilbert space needed to purify a given set of operators $\{h_1, \ldots, h_m\}$.
Specifically, in Sec.\ \ref{sec2} we analyze the case in which one is interested in purifying two linearly independent Hamiltonians. In this context we provide the exact value for $d_E^{\text{(min)}}$ when the input Hilbert space has dimension $d=2$ or $d=3$ and give lower and upper bounds for the remaining configurations. 
 In Sec.\ \ref{sec.4} instead we present a generic construction which allows one to put a bound on $d_E^{\text{(min)}}$ when the set of the operators $\{h_1, \ldots, h_m\}$ contains an arbitrary number $m$ of linearly independent 
elements. In Sec.\ \ref{sec5} we discuss the case in which the total number of linearly independent elements of $\{h_1, \ldots, h_m\}$ 
is maximum, i.e., equal to $d^2$ with $d$ being the dimension of the input Hilbert space. Under this condition we compute the exact value of $d_E^{\text{(min)}}$, showing that it is equal to $d^2$. As we shall see this corresponds to provide a spanning-set  purification of the whole algebra $\mathfrak{u}(d)$ in terms of the largest commutative subalgebra of $\mathfrak{u}(d^2)$. Finally in Sec.\ \ref{sec:Daniel}  we prove that it is always possible to obtain a generator purification of the algebra ${\mathfrak{u}(d)}$ with an extended space of dimension $d_E=d+1$, i.e., in terms of the largest commutative subalgebra of  ${\mathfrak{u}(d+1)}$.
Conclusions and perspectives are given in Sec.\ \ref{sec.con}, and the proof of a Theorem is presented in the Appendix.

\section{Definitions and Basic Properties} \label{sec1}
In this section we start by presenting a rigorous formalization of the problem and discuss some basic properties.
\begin{mydef}[\textbf{Hamiltonian purification}]
Let $\mathcal{S}  := \{ h_1, \ldots, h_m\}$ be a collection 
of $m$ self-adjoint operators (Hamiltonians) acting on a Hilbert space $\mathscr{H}_d$ of dimension $d$. Given then a collection $\mathcal{S}_\mathrm{ext}  :=  \{H_1,  \ldots, H_m\}$ of self-adjoint operators acting on an extended Hilbert space $\mathscr{H}_{d_E}$ which includes $\mathscr{H}_d$ as a proper subspace (i.e., $d_E = \dim\mathscr{H}_{d_E} \geq d$), we say that $\mathcal{S}_\text{ext}$ provides a \textit{purification} for $\mathcal{S}$ if all elements of $\mathcal{S}_\text{ext}$  \textit{commute} with each other, i.e., 
\begin{equation}
[H_j,H_{j'}] =0,   \quad \text{for all $j,j'$,} \label{commm}
\end{equation}
and are related to those of $\mathcal{S}$ as 
\begin{equation}
h_j = P H_j P,  \quad \text{for all $j$,} \label{eq1}
\end{equation}
where $P$ is the orthogonal projection onto $\mathscr{H}_{d}$
\cite{NOTA6}.
\end{mydef}

The requirement (\ref{commm}) that the operators   of   $\mathcal{S}_\text{ext}$ 
 are pairwise commuting implies that such a set spans an Abelian (i.e., commutative) subalgebra  of $\mathfrak{u}(d_E)$, and that
$H_j$ can be simultaneously diagonalized with a single unitary operator $U$ \cite{HJ12}, i.e.,
\begin{equation}\label{diago111} 
H_1 = U D_1 U^\dag, \ldots ,H_m = UD_m U^\dag,
\end{equation}
with $D_1, \ldots, D_m$  being real diagonal matrices.

By construction, it is  clear that each one of the elements of  $\mathcal{S}_\text{ext}$ in general depends upon \textit{all} the operators of 
the set  $\mathcal{S}$  which one wishes  to purify, and not just upon the one it extends.
Furthermore, if $h_j$ satisfy some special relations, identifying $\mathcal{S}_\text{ext}$ may be simpler than in the general case.
For instance, if all the elements of $\mathcal{S}$ admit a set of common eigenvectors,  they already commute in the subspace spanned by those eigenvectors. Then, we are left with the simpler problem of making the operators commute only on the complementary subspace. 
To keep the analysis as general as possible  we will not consider these special cases in the following. We will however make 
use of the linearity of  Eq.\ (\ref{eq1})  to simplify the analysis.

\begin{lem}\label{lemma1}
Let $\mathcal{S}=\{ h_1,  \ldots , h_m\}$ be a collection of self-adjoint operators acting on the Hilbert space $\mathscr{H}_d$ and suppose that a purifying set  $\mathcal{S}_\text{ext}  =  \{H_1, \ldots, H_m\}$ can be constructed on $\mathscr{H}_{d_E}$. Then:
\begin{enumerate}
\item Given $\mathcal{S}'=\{ h'_1, \ldots , h'_{m'}\}$ a collection of self-adjoint operators obtained by taking linear combinations of the elements of $\mathcal{S}$, i.e., 
\begin{equation}
h'_i = \sum_{j=1}^m \alpha_{i,j} h_j
\label{eqn:LinearComb}
\end{equation}
 with $\alpha_{i,j}$ being elements of a real rectangular  $m'\times m$  matrix, then a
 purifying set for $\mathcal{S}'$  on $\mathscr{H}_{d_E}$ is provided by $\mathcal{S}'_\text{ext}  =  \{H'_1,  \ldots, H'_{m'}\}$ with  elements
\begin{equation}
H'_i = \sum_{j=1}^m \alpha_{i,j} H_j  ;
\label{eqn:LinearCombExt}
\end{equation}

\item  Any subset of linearly independent elements of $\mathcal{S}$ corresponds to a subset of linearly independent elements in $\mathcal{S}_\text{ext}$ (the opposite statement being not true in general, i.e., linear independence among the elements of $\mathcal{S}_\text{ext}$ does not imply
  linear independence among the elements of $\mathcal{S}$);

\item For $\lambda_1, \ldots ,\lambda_m \in \mathds{R}$, calling $I_d$ the identity on $\mathscr{H}_d$ and $I_{d_E}$ the identity on $\mathscr{H}_{d_E}$, a purifying set for
\begin{equation}
\{ h_1 + \lambda_1 I_d, \ldots , h_m + \lambda_m I_d\}
\end{equation}
is given by
\begin{equation}
\{H_1 + \lambda_1 I_{d_E}, \ldots , H_m + \lambda_m I_{d_E}\}  ;
\end{equation}

\item For any unitary $U \in \mathcal{U}(d)$, setting $\widetilde{U} = U \oplus I_{d_E-d} \in \mathcal{U}(d_E)$, 
a purifying set for
\begin{equation}
\{U h_1 U^\dagger, \ldots , U h_m U^\dagger\}
\end{equation}
is given by
\begin{equation}
\{\widetilde{U} H_1 \widetilde{U}^\dagger , \ldots , \widetilde{U} H_m \widetilde{U}^\dagger\}  .
\end{equation}
\end{enumerate}
\end{lem}
\begin{proof}
These facts are all trivially verified. 
\end{proof}

Property 1 of Lemma \ref{lemma1} implies that a purifying set $\mathcal{S}_\text{ext}=\{H_1, \ldots, H_m\}$  can be extended by linearity to a purification of any linear combinations of the elements of $\mathcal{S}=\{ h_1, \ldots, h_m\}$. 
Accordingly we can say that the purification of  $\mathcal{S}$  by $\mathcal{S}_\text{ext}$ naturally induces a purification of the algebra spanned by the former by the algebra of the latter (more on this in Sec.\ \ref{SEC:al}).  It is also clear that the fundamental parameter of the Hamitonian purification problem is not the number of elements of $\mathcal{S}$ but instead the maximum number of linearly independent elements which can be found in $\mathcal{S}$. Therefore, without
loss of generality,  in the following  we will assume $m$ to coincide with such a number, i.e., that all the elements of $\mathcal{S}$  are linearly
independent. Then, by Property 2 of Lemma \ref{lemma1} 
also the elements of $\mathcal{S}_\text{ext}$ share the same property. By the same token, also the normalization of the operators $h_j$ can be fixed \textit{a priori}. Property 3 can be used instead to assume that all the elements of $\mathcal{S}$ be traceless (an option which we shall invoke from time to time to simplify the analysis).  Finally Property 4 can be exploited to arbitrarily fix a basis on $\mathscr{H}_d$, e.g., the one which diagonalizes the first element of $\mathcal{S}$.

As we shall see in the following sections the mere possibility of finding a purification for a generic set $\mathcal{S}$ can be easily proved. A less trivial issue  is to determine the \textit{minimal} dimension $d_E^{\text{(min)}}$ of the Hilbert space $\mathscr{H}_{d_E}$ which guarantees the existence of a purifying set for a generic collection $\mathcal{S}$ on $\mathscr{H}_d$. Clearly the value of $d_E^{\text{(min)}}$ will depend on the dimension  $d$ of the Hilbert space $\mathscr{H}_d$ and on the number of (linearly independent) elements $m$ of the set, i.e., $d_E^{\text{(min)}} = d_E^{\text{(min)}}(d,m)$.

By construction it is clear that this quantity cannot be smaller than 
$d$ and than $m$, i.e.,
\begin{equation} 
d_E^{\text{(min)}} \geq \max\{d, m\}. 
\end{equation}  
This is a simple relation which, on one side, follows from the observation 
that $ \mathscr{H}_{d_E}$ being an extension of $\mathscr{H}_d$ must have 
dimension $d_E$ at least as large as $d$. On the other side the inequality 
$d_E^{\text{(min)}} \geq m$ can be verified by exploiting the fact that the diagonal 
$d_E\times d_E$ matrices $D_j$ entering Eq.\ (\ref{diago111}) must be linearly 
independent in order to fulfill Property 2 of Lemma \ref{lemma1}. 
Actually for all non-trivial cases the inequality is strict, resulting in
\begin{equation} 
d_E^{\text{(min)}} \geq \max\{d+1,  m+1\} . \label{trivia} 
\end{equation}  
In fact when the initial Hamiltonians $\{h_1, \ldots ,h_m\}$ do not already 
commute, we need to expand the dimension of the space at least by one, 
obtaining $d_E^{\text{(min)}} \geq d+1$. Moreover the inequality $d_E \geq m+1$ always 
holds, unless the identity $I_d$ lies in the span of 
$\{h_1, \ldots ,h_m\}$. Suppose in fact that we can purify a set of $m$ 
linearly independent Hamiltonians in dimension $m$; then the linear 
span of the $m$ (linearly independent) diagonal matrices $D_j$ in 
Eq.\ (\ref{diago111}) includes also the identity matrix 
$I_{d_E}$. Because for any unitary $U$ we have 
$U I_{d_E} U^\dag = I_{d_E}$, 
the projection of $I_{d_E}$ on $\mathscr{H}_d$ gives 
the identity on that subspace, and in conclusion 
we have that ${I_d}\in \Span(h_1, \ldots ,h_m)$. 
Since this is not true in the general case, we obtain
$d_E^{\text{(min)}} \geq m+1$.

\subsection{Algebra purification} \label{SEC:al}
 As anticipated in the previous section the linearity property of the Hamiltonian purification scheme allows us to introduce the notion of purification of an algebra. Specifically
 there are at least two different possibilities:  
\begin{mydef}[\textbf{Purification(s) of an algebra}] \label{def2}
Let   $\mathfrak{a}$ be a {Lie} algebra of {self-adjoint} operators on $\mathscr{H}_d$. Given a commutative Lie algebra $\mathfrak{A}$ of {self-adjoint} operators on  $\mathscr{H}_{d_E}$ we  say that it provides 
\begin{enumerate}
\item a spanning-set purification   (or simply an algebra purification)   of   $\mathfrak{a}$ when we can provide an  Hamiltonian purification of a spanning set (e.g., a basis) of the latter  in $\mathfrak{A}$;
\item a generator purification of  $\mathfrak{a}$  when we can provide a Hamiltonian purification of a  generating set of the latter  in $\mathfrak{A}$.
\end{enumerate} 
\end{mydef}
The spanning-set purification typically requires the purification of more 
Hamiltonians  than the generator purification. For instance in Sec.\ \ref{sec5}
 we shall see  that the {(optimal)} spanning-set purification of $\mathfrak{u}(d)$ requires   
$\mathfrak{A}$ to be the largest commutative subalgebra of $\mathfrak{u}(d^2)$, while in Sec.\ \ref{sec:Daniel} we shall see the generator purification requires  $\mathfrak{A}$ to be the largest commutative subalgebra of $\mathfrak{u}(d+1)$.
  At the level of quantum control via the Zeno effect, the advantage posed by the spanning-set purification is associated with the fact that, 
  in contrast to the scheme based on generator purification, 
no complicated concatenation of Zeno pulses would be necessary to realize a desired control over a system on $\mathscr{H}_{d}$:
  any unitary operator $e^{-i ht}$ on the latter can in fact be simply obtained as in Eq.\ (\ref{zenosequence}) by choosing $H$ to be the linear combination of 
  commuting Hamiltonians which purifies $h$ on $\mathscr{H}_{d^2}$. 
 On the contrary, in the case of generator purification, first we have to decompose 
 $e^{-i ht}$ into a sequence of pulses of the form $e^{-i h_{j_N}t_N}\cdots  e^{-i h_{j_1}t_1}$ with $h_j$ being taken from the generator sets of
 operators for which we do have a purification. Then each of the pulses $e^{-i h_{j_k}t_k}$ entering the previous  decomposition is realized as in Eq.\ (\ref{zenosequence}) 
 with a proper choice of the purifying Hamiltonians. See Fig.\ \ref{fig:1}
  for a pictorial representation. 
  
 \begin{figure}[t]
    \centering
    \includegraphics[scale=.9]{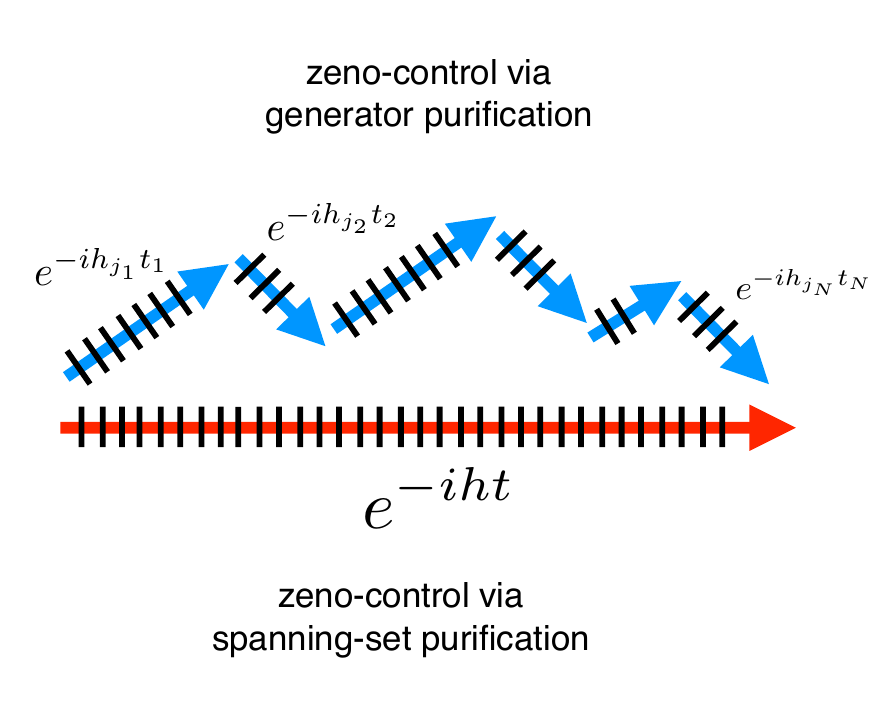}
    \caption{Pictorial representation of the control achieved via spanning-set purification (red line) and generator purification (blue lines) of an algebra. In the the former case an arbitrary unitary transformation  $e^{-i ht}$ on $\mathscr{H}_{d}$ is obtained via a single Zeno sequence (\ref{zenosequence}) with $H$ being the purification of $h$. For generator purification instead one has to use a collection of Zeno sequences, one for each of the generator pulses $e^{-i h_{j_k}t_k}$ which are needed to implement $e^{-iht}$. The black tick lines represent the iterated projections on the system. }
    \label{fig:1}
\end{figure}

\section{Purification of $\bm{m=2}$ Operators}\label{sec2}
In this section we discuss the case of the purification of two linearly independent Hamiltonians (i.e., $m=2$),  providing bounds 
 and exact solutions. In particular we first present a simple construction which shows how to purify into an
 extended Hilbert space $\mathscr{H}_{d_E}$ of dimension $d_E=2d$, implying hence $d_E^{\text{(min)}}(d,m=2)   \leq 2 d$ (Proposition \ref{prop1}).  
 Such a result is interesting because it is elegant and simple to prove. However it is certainly not the optimal. Indeed we will show that the following inequality always holds
 \begin{equation}\label{ris1}
 \left\lceil \frac{3d}{2}\right\rceil \leq d_E^{\text{(min)}}(d,m=2)  \leq 2 d-1.
 \end{equation} 
See Proposition \ref{lower_bound} (lower bound)  and Proposition \ref{purif2ops_2d-1} (upper bound).
For $d=2$ (qubit) and $d=3$ (qutrit)  this allows us to compute exactly $d_E^{\text{(min)}}(d,m=2) $, obtaining 
 respectively
 \begin{align}
&d_E^{\text{(min)}}(d=2,m=2) =3, \quad \text{for a qubit} ,\label{qubit} \\
&d_E^{\text{(min)}}(d=3,m=2) =5, \quad \text{for a qutrit} . \label{qutrit} 
 \end{align}  
For larger values of the system dimension $d$ there is a gap between the lower and upper bounds of the inequality (\ref{ris1}).
Numerical evidences conducted for $d=4,5,6$ however suggest that the former should be always  attainable. See Fig.\ \ref{fig:Graphics_5}.
 \begin{figure}[t]
    \centering
    \includegraphics[scale=.35]{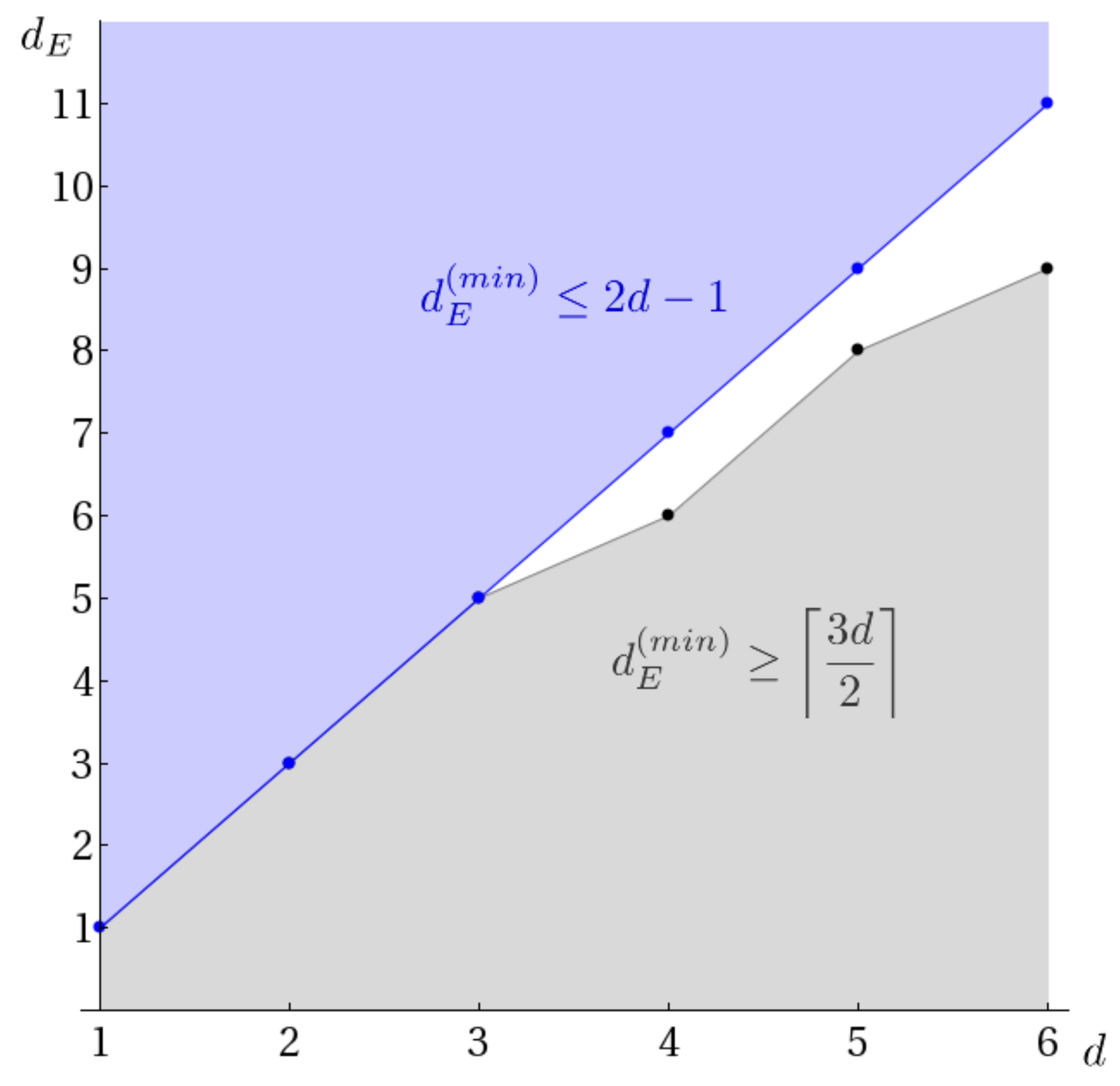}
    \caption{Graphical representation of the bounds (\ref{ris1}) for $d_E^{\text{(min)}}(d,m)$ for $m=2$ as a function of $d$. The blue points give the dimensions for which an explicit construction is known (Propositions \ref{purif_2ops_qubit} and \ref{purif2ops_2d-1}).
    The gray line gives the lower bound of Proposition \ref{lower_bound}, and the black points give the values of $d_E^{\text{(min)}}(d,2)$ estimated by numerical inspection.}
    \label{fig:Graphics_5}
\end{figure} 
\begin{prop}[\textbf{Purification of $\bm{m=2}$ operators with $\bm{d_E = 2  d}$}]\label{prop1}
Let $\mathcal{S}=\{ h_1, h_2\}$ be a collection  of two self-adjoint operators acting on the
Hilbert space $\mathscr{H}_d$. Then a purifying set can be constructed on $\mathscr{H}_{d_E} = \mathscr{H}_d\otimes \mathscr{H}_2$, with $\mathscr{H}_2$ being two-dimensional (qubit space) (i.e., $d_E = 2d$). In particular we can take
\begin{align}
H_1 &= h_1 \otimes I_2 + h_2 \otimes X , \nonumber \\
H_2 &= h_2 \otimes I_2 + h_1 \otimes X , \nonumber \\
P&= I_d \otimes (I_2 + Z)/2,
\end{align} 
where 
$X$ and $Z$ are  the Pauli operators on  $\mathscr{H}_2$~\cite{noteprop1}.
\end{prop}
\begin{proof}
The proof easily follows from the properties of Pauli operators.
But, to get a better intuition on what is going on, it is useful to adopt the following block-matrix representation for $H_j$ and $P$, i.e.,
\begin{equation}
H_1 =\left(  \begin{array}{cc}
h_1 & h_2 \\
h_2 & h_1 \end{array}
\right), \quad 
H_2 =\left(  \begin{array}{cc}
h_2& h_1 \\
h_1 & h_2 \end{array}
\right), \quad 
P =\left(  \begin{array}{cc}
I_d & 0 \\
0 & 0 \end{array}
\right),
\end{equation}
from which the commutativity is evident~\cite{NOTAREMARK}.
\end{proof}

As we shall see in the next section, Proposition \ref{prop1}  
admits a generalization for arbitrary values of $m$. Specifically, independently of the dimension of $\mathscr{H}_d$ (e.g., also for infinite-dimensional systems), we can construct a purification of $m$ not necessarily commuting Hamiltonians, by simply adding an $m$-level system to the original Hilbert space. In the case $\mathscr{H}_d$ is finite-dimensional this implies that a purification for $m$ Hamiltonians can always be achieved with an extended Hilbert space which has at most $m$ times the dimension of the original one, i.e., $d_E =  m d$.  This of course is not the best option. Indeed already for $d=2$ (qubit) and $m = 2$, it is possible to show (see Proposition \ref{purif_2ops_qubit} below) that the purification of two arbitrary Hamiltonians $h_1$ and $h_2$ is attained with a qutrit, i.e., $d_E= d+1= 3$, and this is clearly the optimal solution.

\begin{prop}[\textbf{Optimal purification of $\bm{m=2}$ operators of a qubit}]
\label{purif_2ops_qubit}
Let $\mathcal{S}=\{ h_1, h_2\}$ be a collection composed of  two self-adjoint operators acting on the Hilbert space $\mathscr{H}_2$ of a qubit. Then a purifying set can be constructed on the Hilbert space $\mathscr{H}_{d_E}$  of dimension $d_E=3$ (qutrit space).
\end{prop}
\begin{proof}
We prove the thesis by providing an explicit purification. To do so we first notice that, up to irrelevant additive and renormalization factors, the operators $h_1$ and $h_2$ can be expressed as 
\begin{equation}
h_1 = Z ,\qquad
h_2 = Z + \alpha (X \cos\theta   +Y \sin\theta ) , 
\end{equation} 
with $\alpha$ and $\theta$ being real parameters. 
Indicating then with  $\{ |0\rangle,|1\rangle\}$ the eigenvectors of $Z$, we define 
$\mathscr{H}_{d_E}$ as the space spanned by the vectors $\{ |0\rangle, |1\rangle, |2\rangle\}$ 
with $|2\rangle$ being an extra state which is assumed to be orthogonal to both $|0\rangle$ and $|1\rangle$. We hence introduce the operators on  $\mathscr{H}_{d_E}$ which in the basis 
$\{ |0\rangle, |1\rangle, |2\rangle\}$ have the following matrix form, 
\begin{align}
\widetilde{H}_1
&=\left(
\begin{array}{c|c}
Z&\begin{array}{c}0\\\sqrt{2}\end{array}\\
\hline
\begin{array}{cc}0&\sqrt{2}\end{array}&0 
\end{array}
\right), 
\nonumber\\
\widetilde{H}_2
&=\left(
\begin{array}{c|c}
M&\begin{array}{c}\sqrt{2}\,e^{-i\theta}\\0\end{array}\\
\hline
\begin{array}{cc}\sqrt{2}\,e^{i\theta}&0\end{array}&0 
\end{array}
\right), 
\label{eqn:PurificationPauli}
\end{align}  
with $M :=  X \cos\theta +  Y \sin\theta$. 
One can easily verify that they commute, $[ \widetilde{H}_1, \widetilde{H}_2] =0$, and when projected on  the subspace $\{ |0\rangle, |1\rangle\}$ they yield the matrices $Z$ and $M$, respectively. Defining hence $H_1$ and $H_2$ as the operators
\begin{equation} 
 H_1 = \widetilde{H}_1,\qquad
 H_2 = \widetilde{H}_1 + \alpha \widetilde{H}_2 ,
\end{equation} 
one notices that this is indeed a purifying set of $\mathcal{S}$.
\end{proof}

For arbitrary values of  $d$ an improvement with respect to Proposition \ref{prop1} is obtained as follows:
\begin{prop}[\textbf{Purification of $\bm{m=2}$ operators with $\bm{d_E = 2  d-1}$}]\label{purif2ops_2d-1}
Let $\mathcal{S}=\{ h_1, h_2\}$ be a collection composed of two self-adjoint operators acting on the Hilbert space $\mathscr{H}_d$. Then, a purifying set can be constructed on $\mathscr{H}_{d_E} = \mathscr{H}_{2d-1}$, implying hence $d_E^{\text{(min)}}(d, m=2) \leq 2 d-1$.
\end{prop}
\begin{proof}
According to Eq.\ (\ref{diago111}) to construct a purifying set we have to find a unitary matrix $U \in \mathcal{U}(2d-1)$ such that 
\begin{align}
h_1 = PUD_1 U^\dagger {P}, \nonumber\\
h_2 = PUD_2 U^\dagger {P}  ,
\label{eqh1h2}
\end{align}
with $D_1,D_2 \in \Diag(2d-1)$ being real diagonal matrices of dimension $2d-1$. 
In $\mathscr{H}_{d_E}= \mathscr{H}_{d} \oplus \mathscr{H}_{d-1}$ we can write
\begin{equation}
P = \left(\begin{array}{c|c}
	I_d & 0 \\
	\hline
	0 & 0_{d-1}
	\end{array}\right), \qquad
PU = \left(\begin{array}{c|c}
	L & R \\
	\hline
	0 & 0_{d-1}
\end{array}\right),
\end{equation}
where $L$ is a $d \times d$ matrix, $R$ is a $d \times (d-1)$ matrix, and the rows of $PU$ are orthogonal to each other, $L L^\dagger + R R^\dagger = I_d$, since $PU U^\dag P=P$.
We then write
\begin{equation}
D_1 = \left(\begin{array}{c|c}
	D_1^L & 0 \\
	\hline
	0 & D_1^R
	\end{array}\right)
,\quad
D_2 = \left(\begin{array}{c|c}
	D_2^L & 0 \\
	\hline
	0 & D_2^R
\end{array}\right),
\end{equation}
where $D_1^L, D_2^L$ are diagonal $d \times d$ matrices and $D_1^R, D_2^R$ are diagonal $(d-1) \times (d-1)$ matrices. Then we notice that the equations in \eqref{eqh1h2} are equivalent to
\begin{gather}
h_1 = L D_1^L L^\dagger + R D_1^R R^\dagger, \nonumber\\
h_2 = L D_2^L L^\dagger + R D_2^R R^\dagger ,\nonumber\\
L L^\dagger + R R^\dagger = I_d.
\label{eqComplete}
\end{gather}
To find the purification, we need to solve these equations.

\textit{First equation:} 
we choose without loss of generality $h_1$ to be positive definite: this can be obtained by adding $\alpha  I_d$ with $\alpha > - \min\sigma(h_1)$ [where $\sigma(X)$ denotes the spectrum of $X$].
Then, $\sqrt{h_1}$ is the Hermitian positive-definite matrix such that $(\sqrt{h_1})^2 = h_1$. We also choose
\begin{equation}
L = \frac{1}{\lambda} \sqrt{h_1}\,V,
\end{equation}
where $V$ is an arbitrary unitary matrix, $VV^\dagger=I_d$.
Notice that for any unitary $V$ we have $\lambda^2 L L^\dagger = \sqrt{h_1}\,VV^\dagger \sqrt{h_1} = h_1$. Accordingly  to solve the first of the equations in \eqref{eqComplete} we can simply take $D_1^L = \lambda^2 I_d$ and $D_1^R = 0$.

\textit{Third equation:} recast the third equation in the form
\begin{equation}
	R R^\dagger = I_d - L L^\dagger = I_d -\frac{1}{\lambda^2}h_1  .
\end{equation}
This equation can be solved for $R$ if and only if the right-hand side is a positive semi-definite matrix with non-null kernel. This can be accomplished by choosing $\lambda^2  :=  \max\sigma (h_1)$, so that the smallest eigenvalue of $I_d - \lambda^{-2}h_1$ is equal to zero (this is easily seen in the basis in which $h_1$ is diagonal) \cite{NOTA5}. Explicitly, we can write
\begin{align}
I_d -\frac{1}{\lambda^2} h_1
& = W
\left(\begin{array}{c|c}
D' & 0\\
\hline
0  & 0
\end{array}\right)
W^\dagger \nonumber \\
& = 
\left(\begin{array}{c|c}
W' & 0 
\end{array}\right)
\left(\begin{array}{c|c}
D' & 0 \\
\hline
0  & 0
\end{array}\right)
\left(\begin{array}{c}
W'^\dagger \\
\hline
0
\end{array}\right),
\end{align}
where $W$ and $D'$ are obtained with the spectral theorem and $W'$ is a $d \times (d-1)$ matrix obtained from $W$ deleting its last column. So a solution to the third equation is given by $R = W'\sqrt{D'}$.

\textit{Second equation:} we exploit the fact that $V$ is so far an arbitrary unitary matrix. We take $D_2^R = 0$, and then we are left with
\begin{equation}
h_2 = \frac{1}{\lambda^2} \sqrt{h_1}\,VD_2^L V^\dagger\sqrt{h_1},
\end{equation}
or equivalently
\begin{equation}
\lambda^2 h_1^{-1/2} h_2 h_1^{-1/2}= VD_2^L  V^\dagger ,
\end{equation}
which can be solved for $V$ and $D_2^L$ using the spectral theorem.

In conclusion, the explicit purification of $h_1$ and $h_2$, with $h_1$ positive definite, are found by extending 
\begin{equation}
PU = \left(\begin{array}{c|c}
\lambda^{-1}\sqrt{h_1}\,V&W'\sqrt{D'}\\
\hline
0 & 0_{d-1}
\end{array}\right)
\end{equation} 
to a unitary matrix $U$ and then expressing $H_1$ and $H_2$ as
\begin{equation}
H_1 = 
U
	\left(\begin{array}{c|c}
	\lambda^2 I_d & 0 \\
	\hline
	0  			& 0_{d-1}
	\end{array}\right)
U^\dagger  ,\quad
H_2  = 
U
	\left(\begin{array}{c|c}
	D_2^L & 0 \\
	\hline
	0  	  & 0_{d-1}
	\end{array}\right)
U^\dagger.
\end{equation}
\end{proof}

\begin{prop}[\textbf{Lower bound on the purification of $\bm{m=2}$ operators}]
\label{lower_bound}
The minimum dimension $d_E^{\text{(min)}}$ of the extended space on which it is possible to purify an arbitrary set of two Hamiltonians $\{ h_1, h_2\}$ acting on $\mathscr{H}_d$ is greater or equal to $3d/2$, i.e., $d_E^{\text{(min)}}(d,m=2)  \geq  \lceil 3d/2\rceil$.
\end{prop}
\begin{proof}
We want to find $H_1$ and $H_2$,
\begin{equation}
H_1 = 
\left(\begin{array}{c|c}
h_1&B_1 \\
	\hline
B_1^\dagger&C_1
\end{array}\right), \qquad
H_2 = 
\left(\begin{array}{c|c}
	h_2& B_2 \\
	\hline
	B_2^\dagger& C_2
\end{array}\right),
\end{equation}
such that $[H_1,H_2]=0$. Writing the commutators in block form, we obtain
the following three equations 
\begin{gather}
[h_1,h_2] = - B_1 B_2^\dagger + B_2 B_1^\dagger, \nonumber\\
h_1B_2 - h_2B_1 = - B_1C_2 + B_2C_1 ,\nonumber\\
B_1^\dagger B_2 - B_2^\dagger B_1 = -[C_1, C_2] .
\end{gather}
Actually in order to prove the thesis, we need to consider just the first of these equations. In general $[h_1, h_2]$ can be of maximal rank, i.e., of rank $d$ \cite{NOTE2}. On the other hand $B_1$ and $B_2$ have ranks at most equal to $d_E-d$ (the number of their columns), and so $- B_1 B_2^\dagger + B_2 B_1^\dagger$ has rank at most equal to $2d_E-2d$. Therefore we have to impose $d = \rank(- B_1 B_2^\dagger + B_2 B_1^\dagger) \leq 2d_E - 2d$, which implies $d_E \geq {3}d/2$. 
\end{proof}

For $d=2$ the lower bound of Proposition \ref{lower_bound} is trivial as it only predicts that the minimal value $d_E^{\text{(min)}}$ should be 3 which is the smallest dimension we can hope for to construct a space $\mathscr{H}_{d_E}$ that admits a proper bi-dimensional subspace. 
In Proposition \ref{purif_2ops_qubit} we have explicitly provided a purification for the case  $m=2$ and $d=2$, which uses exactly  $d_E=3$, proving hence that the inequality of Proposition \ref{lower_bound} is tight at least in this case. The same result holds for $d=3$, as it is clear by comparing Proposition \ref{lower_bound} with Proposition \ref{purif2ops_2d-1}, yielding Eq.\ (\ref{qutrit}).

\section{An Upper Bound for $\bm{d_E^{\text{(min)}}(d,m)}$ for  arbitrary $\bm{m}$ and $\bm{d}$}  \label{sec.4}
Here we provide an explicit construction which  generalizes Proposition \ref{prop1} to the case in which $\mathcal{S}$ is composed of $m\geq 2$ linearly independent elements and allows us to prove the following upper bound 
\begin{eqnarray}
d_E^{\text{(min)}}(d,m) \leq m d . \label{ineqimpo}
\end{eqnarray} 
While it is not tight [e.g., see Propositions \ref{purif_2ops_qubit} and \ref{purif2ops_2d-1} as well as Eq.\ (\ref{RISUL1}) below] this bound most likely gives the proper scaling in terms of the parameter $d$ at least for small  values of $m$.

\begin{theo}[\textbf{Purification of $\bm{m}$ operators with $\bm{d_E = md}$}]
\label{theo_1}\label{teo1} Let $\mathcal{S}=\{ h_1, \ldots ,h_m\}$ be a collection of self-adjoint operators acting on the Hilbert space $\mathscr{H}_d$. Then, a purifying set can be constructed on $\mathscr{H}_{d_E} = \mathscr{H}_d\otimes \mathscr{H}_m$, 
implying hence Eq.\ (\ref{ineqimpo}).
\end{theo}
\begin{proof}
We work in a fixed orthonormal basis, in which $\{|e_1\rangle, \ldots , |e_d\rangle\}$ span $\mathscr{H}_d$, $\{|f_1\rangle, \ldots , |f_m\rangle\}$ span $\mathscr{H}_m$, and thus $\{|e_\ell\rangle \otimes |f_i\rangle\}_{\ell \in \{1,\ldots, d\}, i \in \{1,\ldots, m\}}$ span the extended space $\mathscr{H}_{d_E} = \mathscr{H}_d\otimes \mathscr{H}_m$. We then use the spectral theorem to write $h_i = U_i D_iU_i^{\dagger}$, $\forall i$, with $D_i$ and $U_i$ being operators which, in the orthonormal basis $\{|e_1\rangle, \ldots , |e_d\rangle\}$,   are described by 
diagonal and unitary matrices, respectively.  
A purifying set can then be assigned  by  introducing 
the following operator in  $\mathscr{H}_{d_E}$ 
\begin{equation}
W  := \frac{1}{\sqrt{m}} \sum_{i=1}^m U_i \otimes f_{1i},
\end{equation} 
where $f_{ij} := |f_i\rangle\langle f_j|$, $f_i := f_{ii}=|f_i\rangle\langle f_i|$. One gets 
\begin{equation}
W W^\dag = \frac{1}{m} \sum_{i,j=1}^m  U_i U_j^\dag \otimes f_{1i} f_{j1} = I_d \otimes f_{1} =: P.
\end{equation}
Therefore, $W$ is a partial isometry in $\mathscr{H}_{d_E}$ and $P$ is the orthogonal projection onto its range $\mathscr{H}_d\otimes \mathbb{C} |f_1\rangle\cong  \mathscr{H}_d$. Now consider its polar decomposition $W=P U$ for some (non-unique) unitary $U$  on $\mathscr{H}_{d_E}$. [In terms of representative matrices in the canonical basis  the projection $P$ selects the first $d$ rows of an arbitrary $md\times md$ matrix. Therefore,
 since the first $d$ rows of $W$ are orthonormal they can be extended to build up a unitary matrix $U \in \mathcal{U}(md)$, such that $W=PU$].
By explicit computation one can then observe that the following identity holds:
\begin{equation}
h_i \otimes f_1 = PU(m D_i \otimes f_i)U^{\dagger} P.
\end{equation} 
Accordingly the purifying set can be identified with the operators 
 $H_i =U (m D_i \otimes f_i)U^\dag$.
\end{proof}

\section{Optimal purification  of the whole algebra $\bm{(m=d^2)}$}\label{sec5}
In this section we focus on the case where the set  $\mathcal{S}$ one wishes to purify is large enough to span the whole algebra $\mathfrak{u}(d)$  of $\mathscr{H}_d$, i.e., 
according to Definition \ref{def2}, we study the spanning-set purification problem of $\mathfrak{u}(d)$.
 This corresponds to having $m=d^2$ linearly independent elements in $\mathcal{S}$ (the maximum allowed by the dimension of the Hilbert space of the problem). It turns out that for this special case $d_E^{\text{(min)}}$ can be computed exactly showing that it saturates the bound of Eq.\ (\ref{trivia}), i.e., 
\begin{equation} 
d_E^{\text{(min)}}(d, m=d^2) = d^2 .\label{RISUL1}
\end{equation} 
On one hand this incidentally confirms that the bound of Theorem \ref{theo_1} is not thight. On the other hand it shows that a spanning-set purification for $\mathfrak{u}(d)$
requires the largest commutative subalgebra of $\mathfrak{u}(d^2)$ as minimal purifying algebra.

We start by proving this result for the case of  $n$ qubits  (i.e., $d=2^n$), as this special case admits a simple analysis 
(see Proposition \ref{Sigma_l} and  Corollary \ref{tensor_purif}).
The case of arbitrary $d$ is instead discussed in Theorem \ref{theo_2} by presenting a construction which allows one to purify an arbitrary set of $m=d^2$ linearly independent Hamiltonians in an extended Hilbert space of dimension $d^2$. Finally in Theorem \ref{Almost_all} we prove  that the explicit solution proposed in  Theorem \ref{theo_2} is far from being unique.

\begin{prop}[\textbf{Optimal purification of $\bm{\mathfrak{u}(2)}$}]
\label{Sigma_l}
A spanning-set purification for the algebra of  $\mathfrak{u}(2)$
can be constructed on an extended Hilbert space of dimension $d_E=4$, i.e.,  $\mathscr{H}_{d_E}  = \mathscr{H}_4$.
This is the optimal solution. 
\end{prop}
\begin{proof} By Property 3 of Lemma \ref{lemma1} we can restrict the problem to the case of the traceless operators of $\mathscr{H}_2$, i.e., we can focus on the $\mathfrak{su}(2)$ subalgebra. A set of linearly independent elements for such a space is provided by the Pauli matrices   $\{X,Y,Z\}$.
A purifying set $\{\Sigma_x,\Sigma_y,\Sigma_z\}$ of $\{X,Y,Z\}$ on $\mathscr{H}_4$ can then be exhibited explicitly, considering the following $4\times 4$ matrices,
\begin{align}
& \Sigma_x =
\left(
\begin{array}{cc|cc}
 0   & 1   & 1+i & 0 \\
 1   & 0   & 1+i & 0 \\ \hline
 1-i & 1-i & 1   & 0 \\
 0   & 0   & 0   & -1
\end{array}
\right), \nonumber\\
& \Sigma_y =
\left(
\begin{array}{cc|cc}
 0             & -i           & i & \frac{2+4i}{3}\\
 i             & 0            & 1 & \frac{1-i}{3} \\ \hline
 -i            & 1            & 0 & -1            \\
 \frac{2-4i}{3}& \frac{1+i}{3}& -1& 0
\end{array}
\right),\nonumber\\
& \Sigma_z =
\left(
\begin{array}{cc|cc}
 1              & 0              &- \frac{4+4i}{9}& \frac{7+8i}{9} \\
 0              & -1             & \frac{5+5i}{9} & -\frac{16-i}{9}\\ \hline
 -\frac{4-4i}{9}& \frac{5-5i}{9} & 0              & -i             \\
 \frac{7-8i}{9} & -\frac{16+i}{9}& i              & 0
\end{array}
\right), \label{Sigma_l_Mat}
\end{align}
and taking $P= I_2 \otimes(I_2 +Z)/2$.
It can be seen by direct calculation that they indeed commute. The optimality of the solution follows from the inequality (\ref{trivia}). 
\end{proof}

\begin{cor}[\textbf{Optimal purification of $\bm{\mathfrak{u}(2^n)}$}]
\label{tensor_purif}
Consider $\mathfrak{u}(2^n)$, the Lie algebra of self-adjoint operators acting on $n$ qubits (i.e., $\mathscr{H}_d =\mathscr{H}_2^{\otimes n}$). Then, a spanning-set purification  for this algebra can be constructed with operators acting on $\mathscr{H}_{d_E} = \mathscr{H}_4^{\otimes n}$. This is the optimal solution.
\end{cor}
\begin{proof}
This result follows by observing that any element of $\mathfrak{u}(2^n)$ can be expressed as a linear combination of tensor products of $n$ (generalized) Pauli operators $S_\ell$, with the definitions $S_0 = I_2$, $S_1=X$, $S_2=Y$, $S_3=Z$: 
\begin{equation}
h_j = \sum_{\substack{\ell_1,\ldots, \ell_n \\ \in \{ 0,1,2,3\} }}  \beta^{(j)}_{\ell_1, \ldots, \ell_n}S_{\ell_1} \otimes \cdots \otimes S_{\ell_n}\quad
(j=1,\ldots,2^{2n}).
\end{equation} 
Consider then the set formed by the operators 
\begin{equation}
H_j = \sum_{\substack{\ell_1, \ldots, \ell_n \\ \in \{ 0,x, y,z\} }}  \beta^{(j)}_{\ell_1, \ldots, \ell_n} \Sigma_{\ell_1}  \otimes \cdots \otimes \Sigma_{\ell_n},
\end{equation} 
with $\Sigma_\ell$ defined in Eq.\ \eqref{Sigma_l_Mat}. 
The operators $H_j$ act on the Hilbert space $\mathscr{H}_{d_E} = \mathscr{H}_4^{\otimes n}= \mathscr{H}_2^{\otimes 2 n}$ and commute with each other (this is because they are tensor products of commuting elements). Finally, by projecting them with $P= [I_2 \otimes (I_2+Z)/2]^{\otimes n}$ they yield $h_j$. The solution is optimal due to Eq.\ (\ref{trivia}). 
\end{proof}

The above can be used to bound the minimal value of $d_E$ for the case of an arbitrary finite-dimensional system $\mathscr{H}_d$ by simply embedding it into a collection of qubit system. Specifically consider 
 $\mathcal{S}
=\{ h_1, \ldots, h_m\}$,  a collection of $m$ (not necessarily commuting) self-adjoint operators acting on the  Hilbert space $\mathscr{H}_d$ of finite dimension $d$. Then, setting $n_0 = \lceil \log_2 d \rceil$, a purifying set for $\mathcal{S}$ can be constructed on 
$\mathscr{H}_{d_E} = \mathscr{H}_4^{\otimes n_0}$. This implies that  $d_E$ can be chosen to be equal to $4^{n_0}= (2^{n_0})^2\simeq d^2$. As a matter of fact, 
this result can be strengthened by showing that indeed $d_E = d^2$ independently of the dimension $d$.

\begin{theo}[\textbf{Optimal purification of $\bm{\mathfrak{u}(d)}$}]
\label{theo_2}
A spanning-set purification  for $\mathfrak{u}(d)$ can be constructed on $\mathscr{H}_{d_E} = \mathscr{H}_{d^2}$.
This is the optimal solution.
\end{theo}
\begin{proof} 
The proof is given in the Appendix, where a purifying set is explicitly constructed.
\end{proof}

The construction presented in the proof of Theorem \ref{theo_2} in the Appendix provides a matrix $U$ that allows to perform the purification of all the Hermitian matrices in $\mathfrak{u}(d)$. But actually we notice that almost any unitary matrix will do the job equally well, as we show now. So there is almost free choice in determining a matrix $U$ that accomplishes the task, which can even be chosen at random in the parameter space.
\begin{theo}
\label{Almost_all}
Almost all unitary matrices $U \in \mathcal{U}(d^2)$ [with respect to (every absolutely continuous measure with respect to)  Haar measure] are such that the map $f_{PU} $ defined in the proof of Theorem \ref{theo_2} is surjective. This implies that almost all unitary matrices $U \in \mathcal{U}(d^2)$ provide a purification for all sets of Hermitian operators.
\end{theo}
\begin{proof}
The linear application $f_{PU}$ defined in Eq.\ (\ref{defFU}) maps $\Diag(d^2)$ into $\mathfrak{u}(d)$, which are both $d^2$-dimensional real vector spaces, and so it is surjective if and only if its determinant is different from zero.
 Calling $x_{\ell,k}$ the entries of the matrix $U$, we see that $f_{PU} $ depends quadratically on the complex variables $x_{\ell,k}$, and its determinant $\det f_{PU}$ is a polynomial in these variables.

Preliminarily, if we take $U$ to be an arbitrary complex matrix, i.e., not necessarily unitary, the Theorem can be straightforwardly proved. In fact the set of $U$'s which make $f_{PU} $ non-surjective are the zeros of the polynomial $p(u_1,u_2, \ldots)  :=  \det f_{PU}$, where $u_1,u_2, \ldots$ are real parameters 
 which encode the matrix $U$. Such a polynomial is clearly non-vanishing, as we have found in Theorem \ref{theo_2} an instance of  $U$ for which $f_{PU}$ is surjective. The zero set of a non-null analytic function is a closed set (as it is preimage of a closed set), nowhere dense (otherwise the analytic function would be zero on all its connected domain of convergence), and has zero Lebesgue measure. We prove this by induction. The proposition is true for non-null analytic functions of one real variable, as the zero set is discrete. In general, suppose that $g(x_1,x_2,\ldots,x_K)$ is a non-null analytic function of real variables in $\mathds{R}^K$. Then fixing $x_1$, the function $g_{x_1}(x_2,\ldots, x_K):=g(x_1,x_2,\ldots,x_K)$ is an analytic function of $K-1$ variables.
 Calling $S$ and $S(x_1)$ the zero sets of $g(x_1,x_2,\ldots,x_K)$ and $g_{x_1}(x_2, \ldots, x_K)$, respectively, by induction hypothesis $S(x_1)$ must have $(K-1)$-dimensional Lebesgue measure zero, for all except countably many values $x_1 \in \mathds{R}$. Then we integrate the characteristic function
\begin{align}
&\int \mathbf{1}_{S}(x_1, x_2,\ldots,x_K)\, dx_1 dx_2\cdots dx_K  \nonumber \\
	&\qquad
	 = \int \left( \int \mathbf{1}_{S(x_1)} (x_2,\ldots,x_K)\, dx_2 \cdots dx_K \right) dx_1 \nonumber \\
	&\qquad = \int 0\, dx_1\quad\text{(almost everywhere)} \nonumber \displaybreak[0]\\
	&\qquad = 0\vphantom{\int}
\end{align}
to achieve the stated result.

The same argument applies also when we restrict $U$ to be unitary. In fact, any unitary matrix can be obtained as an exponential of a Hermitian matrix. So the same reasoning as above applies to the analytic function $g(h_1,\ldots,h_K) = \det f(e^{iH})$ where $h_1,\ldots, h_K$ are real parameters which encode the Hermitian matrix $H$ [formally, the proof proceeds by considering a set of local charts that cover the manifold $\mathcal{U}(d^2)$]. Moreover, it can be shown that the Haar measure on $\mathcal{U}(d^2)$ is obtained from the Lebesgue measure on $\mathfrak{u}(d^2)$ via multiplication by a Jacobian of an analytic function, which is always regular, and the property of having zero measure is preserved under this operation.
\end{proof}

\section{Generator purification of $\bm{\mathfrak{u}(d)}$ into $\bm{\mathscr{H}_{d+1}}$}
\label{sec:Daniel}
The Propositions in Sec.\ \ref{sec2} concern the purification of two Hamiltonians ($m=2$).
In particular, it was proved in Proposition \ref{purif_2ops_qubit} that two non-commuting Hamiltonians acting on the Hilbert space $\mathscr{H}_2$ of a qubit can be purified into two commuting Hamiltonians in an extended Hilbert space $\mathscr{H}_3$, namely, by extending the Hilbert space \textit{by only one dimension}.
It is in general not the case for a larger system: adding one dimension is typically not enough to purify a couple of Hamiltonians for a system of dimension $d\ge3$, as proved in Proposition \ref{lower_bound}.
See also Eq.\ (\ref{ris1}).

On the other hand, Proposition \ref{purif_2ops_qubit} on the optimal purification for $m=2$ and $d=2$ helps us to prove 
that one can always find a purification of a generating set of $\mathfrak{u}(d)$ which only involves a $d_E=d+1$ dimensional space. Expressed in the language introduced in Definition \ref{def2} this implies that the largest commutative subalgebra of $\mathfrak{u}(d+1)$ provides a generator purification of $\mathfrak{u}(d)$.
More precisely: 
\begin{theo}
\label{Daniel}
A pair of randomly chosen commuting Hamiltonians $H_1$ and $H_2$ on $\mathscr{H}_{d+1}$ almost surely provide a pair of Hamiltonians $h_1$ and $h_2$ which generate the full Lie algebra on $\mathscr{H}_d$, i.e., $\mathfrak{Lie}(h_1,h_2)=\mathfrak{u}(d)$.
In other words, almost all pairs of commuting Hamiltonians in $\mathscr{H}_{d+1}$ are capable of quantum computation in $\mathscr{H}_d$.
\end{theo}
\begin{proof}
To prove this statement, we have only to find an example of such a set $\{H_1,H_2,P\}$ on $\mathscr{H}_{d+1}$ that yields $\{h_1,h_2\}$ generating the full Lie algebra on $\mathscr{H}_d$ (see Ref.\ \cite{Bur14}).
There is a particularly simple pair of generators $\{h_1,h_2\}$ of $\mathfrak{u}(d)$, namely,
\begin{equation}
h_1=\left(
\begin{array}{ccccc}
1&&&&\\
&0&&&\\
&&\ddots&&\\
&&&\ddots&\\
&&&&0
\end{array}
\right),\quad
h_2=\left(
\begin{array}{ccccc}
0&1&&&\\
1&0&1&&\\
&1&0&\ddots&\\
&&\ddots&\ddots&1\\
&&&1&0
\end{array}
\right).
\end{equation}
A proof that these generate $\mathfrak{u}(d)$ is given in Ref.\ \cite{ref:QSI}.
We can purify them in $\mathscr{H}_{d+1}$, by exploiting the formulas presented in Proposition \ref{purif_2ops_qubit} for the purification of a couple of Hamiltonians of a qubit.
Indeed, two $2\times2$ matrices
\begin{equation}
\left(\begin{array}{cc}1&0\\0&0\end{array}\right),\qquad
\left(\begin{array}{cc}0&1\\1&0\end{array}\right)
\end{equation}
are essentially Pauli matrices $Z$ and $X$, and can be purified to
\begin{equation}
\left(\begin{array}{c|cc}1/2&-1/\sqrt{2}&0\\\hline-1/\sqrt{2}&1&0\\0&0&0\end{array}\right),\qquad
\left(\begin{array}{c|cc}0&0&\sqrt{2}\\\hline0&0&1\\\sqrt{2}&1&0\end{array}\right),
\end{equation}
where we have used Properties 1 and 3 of Lemma \ref{lemma1} (multiplication by a constant and shift by the identity matrix) to convert the first matrix into $-(1/2)Z$ and applied the purification formulas in Eq.\ (\ref{eqn:PurificationPauli}), extending the matrices to the top-left by one dimension, instead of to the right-bottom.
This suggests the purification of the above $h_1$ and $h_2$ to
\begin{align}
H_1&=\left(
\begin{array}{c|ccccc}
1/2&-1/\sqrt{2}&0\\\hline
-1/\sqrt{2}&1&0&&&\\
0&0&0&&&\\
&&&\ddots&&\\
&&&&\ddots&\\
&&&&&0
\end{array}
\right),\\
H_2&=\left(
\begin{array}{c|ccccc}
0&0&\sqrt{2}\\\hline
0&0&1&&&\\
\sqrt{2}&1&0&1&&\\
&&1&0&\ddots&\\
&&&\ddots&\ddots&1\\
&&&&1&0
\end{array}
\right).
\end{align}
These matrices actually commute $[H_1,H_2]=0$ and reproduce $h_1$ and $h_2$ once projected by the projection
\begin{equation}
P=\left(
\begin{array}{c|ccc}
0&0&\cdots &0\\\hline
0 &  \\
\vdots &&I_d& \\
0& 
\end{array}
\right).
\end{equation}
The existence of an example makes us sure that all the sets $\{H_1,H_2,P\}$ on $\mathscr{H}_{d+1}$ except for discrete sets of measure zero do the same job, yielding $\{h_1,h_2\}$ generating the full $\mathfrak{u}(d)$ \cite{Bur14}.
\end{proof}

In Ref.\ \cite{Bur14}, it is shown that almost all pairs of commuting Hamiltonians $\{H_1,H_2\}$ of $n$ qubits are turned into $\{h_1,h_2\}$ capable of quantum computation on $n-1$ qubits, by projecting only a single qubit (i.e., $d_E=2^n$ and $d=2^{n-1}=d_E/2$).
The above Theorem \ref{Daniel} shows that the reduction by only one dimension can already make a big difference.

\section{Conclusions}
\label{sec.con}
In this work we have introduced the notion of Hamiltonian purification and the associated notion of algebra purification. As discussed in the Introduction these mathematical properties arise in the context of quantum control induced via a quantum Zeno effect \cite{Bur14}. We focus specifically on the problem of identifying the minimal dimension $d_E^{\text{(min)}}(d,m)$ which is needed in order to purify a generic set of $m$ linearly independent Hamiltonians, providing bounds and exact analytical results in many cases of interest. In particular the value of $d_E^{\text{(min)}}(d,m=d^2)$ has been exactly computed: this corresponds to the case where one wishes  to induce a spanning-set  purification of the whole algebra of operators acting on the input Hilbert space. 
For smaller values of $m$, apart from some special cases discussed in Sec.\ \ref{sec2}, the quantity $d_E^{\text{(min)}}(d,m)$ is still unknown, e.g., see Fig.\ \ref{fig:Graphics_4}, which refers to the case $d=3$. Finally 
for generator purification of $\mathfrak{u}(d)$  we showed that a $(d+1)$-dimensional Hilbert space can be sufficient. This allowed us to strengthen the argument in Ref.\ \cite{Bur14}: a rank-$d$ projection suffices to turn commuting Hamiltonians on the $(d+1)$-dimensional Hilbert space into a universal set in the $d$-dimensional Hilbert space.
\begin{figure}[t]
    \centering
    \includegraphics[scale=.23]{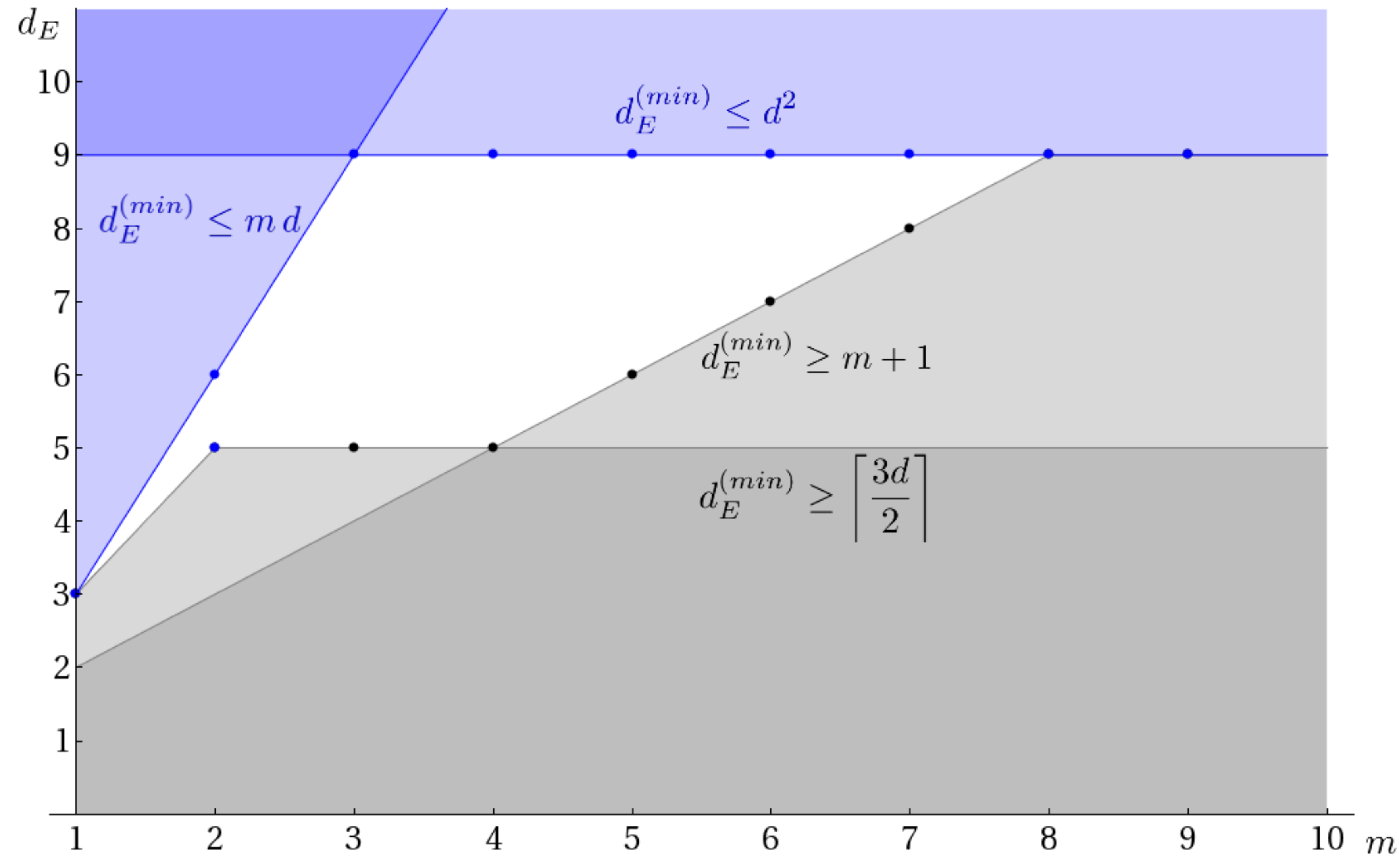}
    \caption{Plots of the admissible regions for  $d_E^{\text{(min)}}(d,m)$ 
for the qutrit case ($d=3$) as functions of $m$. The blue points give the dimensions for which an explicit construction is known.
The gray lines give the known lower bounds on $d_E^{\text{(min)}}(3,m)$.
The black points give the values of $d_E^{\text{(min)}}(3,m)$ estimated by numeric inspection. }
    \label{fig:Graphics_4}
\end{figure}

\acknowledgments
This work was partially supported by the Italian National Group of Mathematical Physics (GNFM-INdAM), by PRIN 2010LLKJBX on ``Collective quantum phenomena: from strongly correlated systems to quantum simulators,'' and by Grants-in-Aid for
Scientific Research (C) (No.\ 22540292 and No.\ 26400406) from JSPS, Japan.

\appendix*
\section{Proof of Theorem \ref{theo_2}}\label{app:ProofTh3}
Here we prove Theorem \ref{theo_2} in Sec.\ \ref{sec5}.
The optimality of the construction follows from the inequality (\ref{trivia}).
From Eq.\ (\ref{diago111}), we can prove that such a solution exists by showing that there are a unitary $U \in \mathcal{U}(d^2)$ and a rank-$d$ projection $P$ defined on $\mathscr{H}_{d^2}$ such that the linear map $f_{PU}:  \Diag(d^2) \to \mathfrak{u}(d)$,
\begin{equation}
   f_{PU} (D ) {\oplus 0_{d^2-d}}   =  P UDU^\dagger {P} , \label{defFU}
\end{equation}
is surjective. Without loss of generality we  are considering $\mathscr{H}_{d^2}= \mathscr{H}_{d} \oplus \mathscr{H}_{d^2-d}$, so that $P=I_d \oplus 0_{d^2-d}$, and~(\ref{defFU}) reads
\begin{equation}
   f_{PU} (D )  =  W D W^\dagger, \label{defFU1}
\end{equation}
where we can parametrize the matrix 
${W}: \mathscr{H}_{d^2} \rightarrow \mathscr{H}_d$ as
\begin{align}
{W} 
&  = 
\left (\begin{array}{ccc|cccc}
	x_{1,1}& \cdots & x_{1,d} & x_{1,d+1} & \cdots & \cdots & x_{1,d^2}\\
	\vdots & \ddots & \vdots    & \vdots & \cdots & \cdots & \vdots \\
	x_{d,1} & \cdots & x_{d,d} & x_{d,d+1} & \cdots & \cdots & x_{d,d^2}
\end{array}\right ) \nonumber \\
& \equiv 
\left( \begin{array}{c}
	X_1 \\
	X_2 \\
	\vdots \\
	X_d
\end{array} \right)
  \equiv 
\left( \begin{array}{ccc|cccc}
	X^1 & \cdots & X^d & X^{d+1} & \cdots & \cdots & X^{d^2} 
\end{array} \right).
\end{align}
{Here}   $x_{\ell,j}$ is the matrix element associated with  the $\ell$th row and the $j$th column of the unitary $U$, and where 
for $\ell\in\{ 1, \ldots, d\}$ we define $X_\ell$ as the complex row vector of $\mathds{C}^{d^2}$ whose $j$th component is 
$x_{\ell,j}$, while  
 for $j\in\{ 1, \ldots, d^2\}$ we define $X^j $ as the complex column vector of  $\mathds{C}^{d}$ whose $\ell$th component 
 is $x_{\ell,j}$.
 The unitarity condition for $U$ requires the row vectors $X_1,\ldots,X_d$ to be orthonormal, i.e.,
 \begin{equation}
 X_\ell \cdot X_{\ell'}^\dag = \sum_{j=1}^{d^2} x_{\ell,j}x^*_{\ell',j} = \delta_{\ell, \ell'} \label{orthoN}.
 \end{equation} 
 The surjectivity condition for $f_{PU} $  instead can be analyzed in terms of  the column vectors $X^j$.
Consider in fact the basis for $\Diag(d^2)$ consisting of matrices $\hat{u}_{ii}$ with $i \in \{1,\ldots,d^2\}$ with only one non-zero entry, 1 in the $i$th position on the diagonal. The function $f_{PU} $ is then surjective if 
the   matrices $f_{PU} (\hat{u}_{11}), \ldots, f_{PU} (\hat{u}_{d^2d^2})$ are linearly independent, i.e., if they span
the whole algebra $\mathfrak{u}(d)$. These are explicitly given by 
\begin{align}
f_{PU} (\hat{u}_{ii}) &  =  {W} \hat{u}_{ii}  {W^\dag}  \nonumber\\
&  = 
\left (\begin{array}{cccc}
	|x_{1,i}|^2     &x_{1,i}x_{2,i}^*& \cdots & x_{1,i}x_{d,i}^* \\
	x_{2,i}x_{1,i}^*& |x_{2,i}|^2    & \cdots & x_{2,i}x_{d,i}^* \\
	\vdots          & \vdots         & \ddots & \vdots           \\
	x_{d,i}x_{1,i}^*&x_{d,i}x_{2,i}^*& \cdots & |x_{d,i}|^2
\end{array}\right )\nonumber\\
& = X^i \times X^{i\dag}, \label{identity1}
\end{align}
where the last identity stresses the fact that, by construction, $f_{PU} (\hat{u}_{ii})$ can be seen as the outer product ``$\times$'' of the vector $X^i$ with itself \cite{NOTA3}.

In order to identify a solution for the problem we have hence to find an assignment for the coefficients $x_{\ell,j}$ which 
 fulfill the condition (\ref{orthoN}) while ensuring that the matrices (\ref{identity1})  span the whole $\mathfrak{u}(d)$.
To show this we proceed by steps. First we identify values for $x_{\ell,i}$ in such a way that the associated $d\times d^2$ matrix  ${\widetilde{W}}$ guarantees that  $\{ {\widetilde{W}} \hat{u}_{ii} {\widetilde{W}^\dagger}\}_{i \in\{1, \ldots, d^2\}}$ provides a basis for $\mathfrak{u}(d)$, hence that the associated mapping {$f_{\widetilde{W}}$} 
is surjective. Then we modify {$\widetilde{W}$} in such a way that the condition (\ref{orthoN}) is fulfilled by orthonormalizing its rows, while making sure that the surjectivity condition of the associated mapping is preserved.

Calling $e_i$ the row vector of $\mathds{C}^d$ with 1 in the $i$th position and introducing $e_{n,m}^{(+)}  :=  e_n + e_m$ and $e_{n,m}^{(-)}  :=  e_n - i e_m$, a basis for $\mathfrak{u}(d)$ is given by the following matrices \cite{NOTA4}
\begin{equation}
\label{decomposition}
\left\lbrace\begin{array}{r@{}l@{\quad}r}
		e_n^\dagger &{}\times  e_n,  &  n \in \{1,\ldots,d\}  , \\
		e_{n,m}^{(+)\dagger} &{}\times e_{n,m}^{(+)},
		&  n <m   \in \{1,\ldots,d\}  , \\
		e_{n,m}^{(-)\dagger} &{}\times e_{n,m}^{(-)},
		& n <m  \in \{1,\ldots,d\} .
\end{array}\right.
\end{equation}
From Eq.\ (\ref{identity1}) it follows that this set  can be obtained as {$f_{\widetilde{W}}(\hat{u}_{ii})= \widetilde{W}\hat{u}_{ii} \widetilde{W}^\dagger$}
if we take as matrix {$\widetilde{W}$} the one with column vectors 
\begin{align}
&
\left\lbrace\begin{array}{l@{}c@{}l}
	X^1  &{}={}&  e_1^\dagger, \\
	 &\vdots& \\
	X^d  &{}={}&  e_d^\dagger ,
\end{array}\right.
\displaybreak[0]\\
&
\left\lbrace\begin{array}{l@{}c@{}l}
	X^{d+1} &{} ={}&  e_{1,2}^{(+)\dagger}, \\
	& \vdots& \\
	X^{2d} & {}={}& e_{1,d}^{(+)\dagger}, \\
	X^{2d+1} &{} = {}&e_{2,3}^{(+)\dagger}, \\
	& \vdots& \\
	X^{\frac{d(d+1)}{2}} &{} ={}& e_{d-1,d}^{(+)\dagger},
\end{array}\right.
\quad
\left\lbrace\begin{array}{l@{}c@{}l}
	X^{\frac{d(d+1)}{2}+1} &{} ={}& e_{1,2}^{(-)\dagger} ,\\
	& \vdots &\\
	X^{\frac{d(d+3)}{2}} & {}={}&  e_{1,d}^{(-)\dagger} ,\\
	X^{\frac{d(d+3)}{2}+1} & {}= {}& e_{2,3}^{(-)\dagger}, \\
	& \vdots &\\
	X^{d^2} & {}= {}&e_{d-1,d}^{(-)\dagger}.
\end{array}\right.
\end{align}
 For instance, in the case $d=4$, this choice gives 
\begin{equation}
\label{PU}
{\widetilde{W}} =
\left(\begin{array}{cccc|cccccccccccc}
 1 & 0 & 0 & 0 & 1 & 1 & 1 & 0 & 0 & 0 & 1 & 1 & 1 & 0 & 0 & 0 \\
 0 & 1 & 0 & 0 & 1 & 0 & 0 & 1 & 1 & 0 & i & 0 & 0 & 1 & 1 & 0 \\
 0 & 0 & 1 & 0 & 0 & 1 & 0 & 1 & 0 & 1 & 0 & i & 0 & i & 0 & 1 \\
 0 & 0 & 0 & 1 & 0 & 0 & 1 & 0 & 1 & 1 & 0 & 0 & i & 0 & i & i
\end{array}\right) .
\end{equation}
Accordingly {$f_{\widetilde{W}}(\hat{u}_{ii})= \widetilde{W}\hat{u}_{ii} \widetilde{W}^\dagger$} span all $\mathfrak{u}(d)$ and so {$f_{\widetilde{W}}$} is a surjective (hence invertible) linear function. Now, {$\widetilde{W}$} does not have orthonormal rows, so it cannot be straightforwardly extended to a unitary operator on $\mathscr{H}_{d^2}$: we have to orthonormalize them. We observe that the scalar product between the rows $X_1,\ldots, X_d$ of {$\widetilde{W}$} gives
\begin{equation}
\left\{
\begin{array}{l@{}c@{}l@{\ \ }r}
	X_n\cdot X_n^\dag &{}={}& 2d-1,    & n   \in \{1,\ldots,d\} ,\\
	X_n\cdot X_m^\dag &{}={}& 1-i,     & n < m  \in \{1,\ldots,d\} ,\\
	X_n \cdot X_m^\dag &{}={}& 1+i,     & n > m  \in \{1,\ldots,d\} .
\end{array}
\right.
\end{equation}
We can orthogonalize them by changing only the entries of the leftmost $d \times d$ submatrix of {$\widetilde{W}$}. In the case $d=4$ we start with
\begin{equation}
A_{(0)} =
\left(\begin{array}{cccc}
 1 & 0 & 0 & 0 \\
 0 & 1 & 0 & 0 \\
 0 & 0 & 1 & 0 \\
 0 & 0 & 0 & 1 
\end{array}\right).
\end{equation}
Then, we make the first row orthogonal to all the others by adding $-1-i$ to all subdiagonal elements in the first column,
\begin{equation}
A_{(1)} =
\left(\begin{array}{cccc}
 1 & 0 & 0 & 0 \\
 -1-i & 1 & 0 & 0 \\
 -1-i & 0 & 1 & 0 \\
 -1-i & 0 & 0 & 1 
\end{array}\right).
\end{equation}
Now, $X_2\cdot X_3^\dag=X_2\cdot X_4^\dag=1-i+(-1-i)(-1+i)= 3-i$, so we can make $X_2$ orthogonal to all the other rows with
\begin{equation}
A_{(2)} =
\left(\begin{array}{cccc}
 1 & 0 & 0 & 0 \\
 -1-i & 1 & 0 & 0 \\
 -1-i & -3-i & 1 & 0 \\
 -1-i & -3-i & 0 & 1 
\end{array}\right).
\end{equation}
Finally, $X_3\cdot X_4^\dag=1-i+(-1-i)(-1+i)+(-3-i)(-3+i)=13-i$, and we can make all the vectors orthogonal with
\begin{equation}
A_{(3)} =
\left(\begin{array}{cccc}
 1 & 0 & 0 & 0 \\
 -1-i & 1 & 0 & 0 \\
 -1-i & -3-i & 1 & 0 \\
 -1-i & -3-i & -13-i & 1 
\end{array}\right).
\end{equation}
This can be extended to any dimension $d$ replacing the leftmost $d\times d$ matrix of {$\widetilde{W}$} with the 
triangular matrix
\begin{equation}
A_{(d-1)} =
\left(\begin{array}{cccccc}
 1 & 0 & 0 & 0 & \cdots & 0 \\
 a_1& 1 & 0& 0 & \cdots & 0 \\
 a_1 & a_2 & 1 &0 & \cdots & 0 \\
  a_1 & a_2 & a_3 & 1 & \cdots & 0 \\
 \vdots & \vdots & \vdots &\vdots &\ddots & \vdots  \\
 a_1 & a_2& a_3 & a_4& \cdots &1 
\end{array}\right),
\end{equation}
where $a_1= -1 -i$ while for $n\in \{ 2,\ldots, d-1\}$  the remaining subdiagonal elements are obtained by solving the recursive equation 
\begin{equation}
	a_n = -1-i - \sum_{k=1}^{n-1}|a_k|^2 .
\end{equation}
For future reference we notice that all $a_n$ have negative real and imaginary parts,
\begin{equation} 
\Re a_n = -\left(1+\sum_{k=1}^{n-1}|a_k|^2\right), \qquad
\Im a_n = -1. \label{ALLPO}
\end{equation} 
Next, the rows of the submatrix $A_{(d-1)}$ are normalized to $1$ obtaining 
\begin{equation}
\left(\begin{array}{cccccc}
 1 & 0 & 0 &  \cdots & 0 \\
 a_1/N_1& 1/N_1 & 0&  \cdots & 0 \\
 a_1/N_2 & a_2/N_2 & 1/N_2 & \cdots & 0 \\
  a_1/N_3 & a_2/N_3 & a_3/N_3 & \ddots & 0 \\
 \vdots & \vdots & \vdots & &  \vdots  \\
 a_1/N_{d-1} & a_2/N_{d-1}& a_3/N_{d-1} & \cdots &1 /N_{d-1}
\end{array}\right),
\end{equation}
with $N_n= \sqrt{1 + \sum_{k=1}^{n} |a_k|^2} = \sqrt{|{\Re a_{n+1}}|}$. We now replace this into
the {$\widetilde{W}$} and normalize the resulting rows to 1 by dividing them by the constants $\sqrt{2d-1}$. 
The resulting $d\times d^2$ matrix is our solution {$W$}. By construction it has orthonormal rows as required by (\ref{orthoN}), 
so it can be extended to a unitary matrix $U$, {such that $W=PU$}.

Moreover the associated function $f_{PU} $ is still surjective. 
This can be proven by induction.
To this end we find it useful to introduce the notion of $k$-submatrix: 
specifically a $k$-lower-right submatrix  ($k$-LRS) is a $d\times d$ Hermitian  matrix whose non-zero entries are only in lower-right submatrix associated with the last $k$ rows and columns. We then call $R$ the right part of the matrix {$W$} (the last $d^2-d$ columns), which is the same as the one we had for the {$\widetilde{W}$} apart from the global rescaling by the factor $1/\sqrt{2d-1}$.
The basic step is to show that, under outer products with themselves, $X^d$ and the columns of $R$ span all the $1$-LRS\@.
This is obvious, as such matrices are obtained as a multiple of $X^d \cdot X^{d\dag}$. Then, we have to show that, if we have $X^d, X^{d-1},\ldots, X^{d-k}$ and the columns of $R$, we can span all $(k+1)$-LRSs. By induction hypothesis, we suppose that we can already obtain all $k$-LRSs. To prove the thesis is then sufficient to show that we can generate the  set of $(k+1)$-LRSs whose non-zero elements are given by 
\begin{align}
\left(\begin{array}{c|ccc}
 1     & 0 &\cdots & 0 \\
 \hline
 0     &   &       &   \\
\vdots &   & \sharp  &   \\
 0     &   &       & 
\end{array}\right),
\left(\begin{array}{c|ccc}
 0     & 1 &\cdots & 0 \\
 \hline
 1     &   &       &   \\
\vdots &   & \sharp  &   \\
 0     &   &       & 
\end{array}\right),
\ldots,
\left(\begin{array}{c|ccc}
 0     & 0 &\cdots & 1 \\
 \hline
 0     &   &       &   \\
\vdots &   & \sharp  &   \\
 1     &   &       & 
\end{array}\right),&\nonumber\displaybreak[0]\\
\left(\begin{array}{c|ccc}
 0     &-i &\cdots & 0 \\
 \hline
 i     &   &       &   \\
\vdots &   & \sharp  &   \\
 0     &   &       & 
\end{array}\right),
\ldots ,
\left(\begin{array}{c|ccc}
 0     & 0 &\cdots &-i \\
 \hline
 0     &   &       &   \\
\vdots &   & \sharp  &   \\
 i     &   &       & 
\end{array}\right),
&
\end{align}
where the symbol ``$\sharp$" represents a generic $k\times k$ matrix.
To achieve this we are allowed to use 
arbitrary linear combinations of the following set of $(k+1)$-LRSs, which are trivially generated via outer product by the vectors $X^d, X^{d-1},\ldots, X^{d-k}$  and by the columns of $R$,
\begin{align}
&\left(\begin{array}{c|ccc}
 1/N_{d-k-1}  &a_{d-k}^*/N_{d-k} &\cdots &a_{d-k}^*/N_{d-1}\\
 \hline
 a_{d-k}/N_{d-k} &   &       &   \\
\vdots &   & \sharp  &   \\
a_{d-k}/N_{d-1}&   &       & 
\end{array}\right),\label{mat_1}\\ 
&
\left(\begin{array}{c|ccc}
 1     & 1 &\cdots & 0 \\
 \hline
 1     &   &       &   \\
\vdots &   & \sharp  &   \\
 0     &   &       & 
\end{array}\right),
\ldots ,
\left(\begin{array}{c|ccc}
 1     & 0 &\cdots & 1 \\
 \hline
 0     &   &       &   \\
\vdots &   &\sharp  &   \\
 1     &   &       & 
\end{array}\right),\label{mat_2}\\ 
&
\left(\begin{array}{c|ccc}
 1     &-i &\cdots & 0 \\
 \hline
 i     &   &       &   \\ 
\vdots &   & \sharp  &   \\
 0     &   &       & 
\end{array}\right),
\ldots ,
\left(\begin{array}{c|ccc}
 1     & 0 &\cdots &-i \\
 \hline
 0     &   &       &   \\
\vdots &   & \sharp  &   \\
 i     &   &       & 
\end{array}\right). \label{mat_3}
\end{align}
The result can then be trivially proved by showing that among such linear combinations one can identify the $(k+1)$-LRSs whose non-zero elements are in the form
\begin{equation}
\left(\begin{array}{c|ccc}
c     &0 &\cdots &0\\
 \hline
 0 &   &       &   \\
\vdots &   & \sharp  &   \\
0&   &       & 
\end{array}\right), \label{target} 
\end{equation} 
with $c\neq 0$. This is done by starting from the matrix  \eqref{mat_1} and then subtracting the off-diagonal elements using the matrices \eqref{mat_2} and \eqref{mat_3}. As a result, we get a matrix (\ref{target}) with 
\begin{equation}
c = \frac{1}{N_{d-k-1}}  - \sum_{j=1}^k\frac{1}{N_{d-j}}(\Re a_{d-k} + \Im a_{d-k}),
\end{equation}
which is indeed different from zero, as according to Eq.\ (\ref{ALLPO}) all the terms are positive.  This concludes the induction step, and the Theorem is proven. {\qed}


\begin{thebibliography}{10}
\bibitem{NOTA1}
Actually one should work in an \textit{interaction picture}, in which the
  system undergoes a spontaneous unitary evolution given by a free Hamiltonian
  $H_\text{free}$, to which an interaction is added as prescribed by the
  experimenter \cite{DH08}. To simplify the discussion we will assume that this
  step has already been factored out, i.e., we imagine that $H_\text{free} =
  0$.
\bibitem{NOTA2}
As a matter of fact, since global phases are irrelevant in quantum mechanics,
  it would be sufficient to focus on the algebra $\mathfrak{su}(d)$ formed by
  the traceless self-adjoint $d\times d$ complex matrices.
\bibitem{NOTA6}
More precisely, $P=I_d \oplus 0$ is a rank-$d$ orthogonal projection 
acting on the Hilbert space $\mathscr{H}_{d_E}=\mathscr{H}_{d} \oplus \mathscr{H}_{d}^\perp$, and 
    Eq.\ (\ref{eq1}) should read $h_j\oplus 0 = P H_j P$. By abuse of notation we will
   use the same symbol for $h_j$ on $\mathscr{H}_d$ and for its extension $h_j\oplus 0$ on
  $\mathscr{H}_{d_E}$. Moreover,  more generally, we can consider a Hamiltonian purification in a space $\mathscr{H}_{d_E}=\widetilde{\mathscr{H}}_{d} \oplus \widetilde{\mathscr{H}}_{d}^\perp$, where $\widetilde{\mathscr{H}}_{d}$ is isomorphic to $\mathscr{H}_d$.
 Again, in the following we will  not be pedantic in  distinguishing isomorphic spaces and will commit the sin of denoting them with the same symbols.
\bibitem{noteprop1}
Notice that this provides a Hamiltonian purification up to the identification $\mathscr{H}_{d} \cong \mathscr{H}_d \otimes \mathbb{C} |0\rangle$, with the subspace $\mathscr{H}_d \otimes \mathbb{C} |0\rangle \subset \mathscr{H}_{d_E}$, and the extension of $h_j$ in $\mathscr{H}_{d_E}$ is $h_j\otimes |0\rangle \langle 0 |$. See~\cite{NOTA6}.
\bibitem{NOTAREMARK}
It is worth observing that the set of operators defined by
\begin{align}
H_1 &= \frac{1}{2}(h_1 \otimes I_2 + h_2 \otimes X ), \nonumber \\
H_2 &=  \frac{1}{2}(h_2 \otimes I_2 + h_1 \otimes X )
\end{align} 
are still commuting and allow one to recover $h_1$ and $h_2$ by simply tracing out the ancilla system $\mathscr{H}_2$. This is a sort of purification of $\mathcal{S}$ where the ancilla system is simply forgotten.
\bibitem{NOTA5}
We assume that $\max\sigma(h_1)$ is not degenerated, but it is not a
  restrictive assumption: if it is degenerated one can purify the operators in
  a Hilbert space with a smaller dimension $d_E$.
\bibitem{NOTE2}
Any anti-Hermitian traceless operator $A$ can be written as the commutator of
  two Hermitian matrices. $A$ in general is unitarily equivalent to an
  anti-Hermitian operator with zero diagonal. Hence, assume, without loss of
  generality, that $A$ has zero diagonal elements. Take $h_1$ diagonal with
  real and distinct diagonal entries $\lambda_1,\ldots,\lambda_d$. The equation
  $A=[h_1,h_2]$ then is equivalent to $(\lambda_i - \lambda_j)h_{ij}^{(2)}=
  a_{ij}$, which has a solution $h_{ij}^{(2)}=a_{ij}/(\lambda_i-\lambda_j)$.
\bibitem{NOTA3}
The identity (\ref{identity1}) clarifies that the $f_{PU} (\hat{u}_{ii})$
  cannot be all mutually orthogonal with respect to the Hilbert-Schmidt scalar
  product. Indeed there can be at most $d$ such matrices which are orthogonal.
  To see this simply observe that $\Tr[ f_{PU} (\hat{u}_{ii})
  f^\dag_{PU}(\hat{u}_{jj})] =|X^{j\dag} \cdot {X^i}|^2$ and remind that the
  $X^i$ are $d^2$ column vectors of $\mathds{C}^d$.
\bibitem{NOTA4}
To see this simply observe that for $n<m$, the matrices $e_n^\dagger \times e_n
  \pm e_m^\dagger \times e_m$, $e_{n,m}^{(+)\dagger} \times e_{n,m}^{(+)}
  -e_n^\dagger \times e_n-e_m^\dagger \times e_m$ and $e_{n,m}^{(-)\dagger}
  \times e_{n,m}^{(-)} -e_n^\dagger\times e_n-e_m^\dagger \times e_m$ form the
  set of generalized Pauli operators in the subspace spanned by the vectors
  $e_n$ and $e_m$.
\end{thebibliography}
\end{document}